\documentclass[a4paper,10pt]{article}
\usepackage[T1]{fontenc}
\usepackage[english]{babel}
\usepackage{amssymb,amsmath,latexsym,epsfig}
\usepackage[all]{xy}
\xyoption{all}
\usepackage[latin1]{inputenc}
\usepackage[T1]{fontenc}
\usepackage[refpage]{nomencl}

\title{Stanley-Reisner resolution of constant weight linear codes\thanks{The original publication is available at http://link.springer.com/article/10.1007/s10623-012-9767-2}}
\author{Trygve Johnsen\thanks{ Dept. of Mathematics, University of Troms{\o}, N-9037 Troms{\o}, Norway, \texttt{Trygve.Johnsen@uit.no}} \and Hugues Verdure\thanks{ Dept. of Mathematics, University of Troms{\o}, N-9037 Troms{\o}, Norway, \texttt{Hugues.Verdure@uit.no}} }

\newcommand{\Fq}{\mathbb{F}_q}
\newcommand{\N}{\mathbb{N}}
\newcommand{\K}{\mathbb{K}}

\newtheorem{remark}{Remark}[section]
\newtheorem{lemma}{Lemma}[section]
\newtheorem{corollary}{Corollary}[section]
\newtheorem{proposition}{Proposition}[section]
\newtheorem{theorem}{Theorem}[section]
\newtheorem{example}{Example}[section]

\newenvironment{proof}[1][Proof]{\begin{trivlist}
\item[\hskip \labelsep {\bfseries #1}]}{\end{trivlist}}

\newcommand{\qed}{\nobreak \ifvmode \relax \else
      \ifdim\lastskip<1.5em \hskip-\lastskip
      \hskip1.5em plus0em minus0.5em \fi \nobreak
      \vrule height0.75em width0.5em depth0.25em\fi}

\begin{document}

\maketitle

\begin{abstract}
\noindent 
Given a constant weight linear code, we investigate its weight hierarchy and the Stanley-Reisner resolution of its associated matroid regarded as a simplicial complex. We also exhibit conditions on the higher weights sufficient to conclude that the code is of constant weight.\\
\noindent
Keywords: constant weight, linear code, Stanley-Reisner resolution, Betti numbers\\
\noindent
2010 Mathematics Subject Classification. 94B05 (05E45, 05B35, 13F55)
\end{abstract}

%
%

\section{Introduction and notation}

In~\cite{Liu1} one found the hierarchy of higher Hamming weights 
for constant weight linear codes over a finite field $\Fq$, and one also found some sufficient conditions to conclude that a linear code is of constant weight (if it is of constant weight of some "higher order"). In the present paper we will first give some other sufficient conditions. Then we will proceed to give more refined information about constant weight  codes by studying the associated matroids derived from parity check matrices. From a more abstract perspective, since~\cite{Eagon1}, it is well known that if one regards a matroid as a simplicial complex (using independent sets as faces), then the Stanley-Reisner ideal of its Alexander dual has a pure, linear $\N$-graded resolution. Furthermore it is clear from Corollary 4.4 of~\cite{Johnsen1} that the Stanley-Reisner ideal of a matroid itself has a pure, linear resolution  if and only its restriction to its set of non-isthmuses  is uniform.
  Here we will exhibit finite matroids (those derived from constant weight codes), 
which themselves have pure $\N$-graded resolution of their Stanley-Reisner rings, but being far from linear. Before giving more details (at the end of this section) about our results  we will explain our notation and concepts. 

Let $\Fq$ be a finite field. A linear $q$-ary code $C$ is a linear subspace of $\Fq^n$ for some $n \in \N$. We denote by $k$ the dimension of the code as a vector space over $\Fq$. A codeword $c$ is an element of the code, and a subcode is a linear subspace of $C$. We denote by $\mathcal{C}_i(C)$ the set of subcodes of dimension $i$ of $C$. Let $c=(c_1,\cdots,c_n)$  be a codeword and $x \in \{1,\cdots,n\}$. We will sometimes write $c|_x$ for $c_x$. Let $D \subset C$ a subcode. Its support is \[Supp(D)=\left\{ x \in \{1,\cdots,n\},\ \exists d \in D,\ d|_x \neq 0\right\}\] and its weight is \[w(D)=\#Supp(D).\] The weight of a codeword is the weight of the subcode generated by it. The minimum (Hamming) distance $d$ of a code is the minimal weight of its non-zero codewords or equivalently of its $1$-dimensional subcodes.  A $[n,k,d]_q$ code is a linear $q$-ary code in $\Fq^n$ of dimension $k$ and minimum distance $d$. \\
In~\cite{Helleseth1}, one generalizes the minimum distance to subcodes of higher dimension, namely, for $1 \leqslant i \leqslant k$, \[d_i=min\left\{w(D),\ D \in \mathcal{C}_i(C)\right\}.\] In particular, $d_1=d$.\\
For our purpose, a code can be given in two equivalent ways: either by a generator matrix or a parity check matrix. A generator matrix $G_C$ of the code $C$ is a $k \times n$ matrix whose row space is $C$. A parity check matrix $H_C$ of the code $C$ is a $(n-k) \times n$ matrix such that \[c=\begin{bmatrix}c_1 & \cdots & c_n\end{bmatrix} \in C \Leftrightarrow cH_C^t = \begin{bmatrix}0 & \cdots & 0\end{bmatrix}.\] Such matrices are not unique for a given code.\\
A constant weight code is a code whose non-zero codewords have the same weight $d$.\\

We refer to~\cite{1Oxley} for the theory of matroids. A matroid $\Delta$ on the set $E=\{1,\cdots,n\}$ can be characterized by many equivalent definitions. We give one here: a matroid is defined by its set $\mathcal{B} \subset 2^E$ of bases satisfying the following properties: \begin{itemize} \item $\mathcal{B} \neq \emptyset$, \item $\forall B_1,B_2 \in \mathcal{B}, \forall x \in B_2-B_1,\ \exists y \in B_1-B_2 \textrm{ such that } B_2-\{x\} \cup \{y\} \in \mathcal{B}$. \end{itemize} An independent set is a subset of a basis, and a circuit is a minimal dependent set. For any subset $\sigma \subset E$, the rank of $\sigma$ is \[rank(\sigma) = max\{\#(B \cap \sigma),\ B \in \mathcal{B}\},\] and for any $1 \leqslant i \leqslant n- rank(E)$, define the higher weights of the matroid by \[d_i = min\{\#\sigma,\ \#\sigma-rank(\sigma) = i\}.\]

Let $C$ be a $[n,k,d]_q$ code given by a parity check matrix $H_C$. We can define a matroid $\Delta(H_C)$ in the following way: its ground set is $E=\{1,\cdots,n\}$ (the indices of the columns of $H_C$) and its set $\mathcal{B}$ of bases is \[\mathcal{B} = \left\{ \begin{array}{l}\sigma \subset E,\ \sigma \textrm{ maximal such that the columns of }H_C\\\textrm{ labelled by }\sigma\textrm{ are linearly independant}\end{array}\right\}.\] It can be shown that (for example in~\cite{Martin1}): \begin{itemize} \item two different parity check matrices give the same matroid,
 \item the rank of the matroid is $n-k$, \item the two sets of $d_i$ defined in this section coincide (those for the code $C$ and those for the matroid $\Delta(H_C)$).\end{itemize}

A simplicial complex $\Delta$ on the finite ground set $E$ is a subset of $2^E$ closed under taking subsets. We refer to~\cite{1Miller} for a brief introduction of the theory of simplicial complexes, and we follow their notation. A matroid is in a natural way a simplicial complex through its set of independent sets. Given a simplicial complex $\Delta$ on the ground set $E$, define its Stanley-Reisner ideal and ring in the following way: let $\K$ be a field and let $S=\K[\mathbf{x}]$ be the polynomial ring over $\K$ in $\#E$ indeterminates $\mathbf{x} = \{x_e,\ e \in E\}$. Then the Stanley-Reisner ideal $I_\Delta$ of $\Delta$ is   \[I_\Delta = <\mathbf{x}^\sigma,\ \sigma \not \in \Delta>\] and its Stanley-Reisner ring is $R_\Delta=S/I_\Delta$. This ring has a minimal free resolution as a $\N^{E}$-graded module \begin{equation} \label{SR} 0 \longleftarrow R_\Delta  \overset{\partial_0}{\longleftarrow} P_0  \overset{\partial_1}{\longleftarrow} P_1 \longleftarrow \cdots \overset{\partial_l}{\longleftarrow} P_l\longleftarrow 0 \end{equation} where each $P_i$ is of the form \[P_i = \bigoplus_{\alpha \in \N^{E}}S(-\alpha)^{\beta_{i,\alpha}}\] and $S(-\alpha)$ is the free module generated in degree $\alpha$, that is $S(-\alpha) \cong <\mathbf{x}^\alpha>$ as $\N^E$-graded modules. Here, $P_0=S$. 
The $\beta_{i,\alpha}$ are called the $\N^{E}$-graded Betti numbers of $\Delta$. We have $\beta_{i,\alpha}=0$ if $\alpha \in \N^E - \{0,1\}^E$. The Betti numbers are independent of the choice of the minimal free resolution, and for matroids, are also independent of the chosen field $\K$ (\cite{Bjorner1}). We can also look at $R_\Delta$ as a $\N$-graded module or an ungraded module. The $\N$-graded and ungraded Betti numbers of $\Delta$ are then respectively the \[\beta_{i,d} = \sum_{|\alpha| = d}\beta_{i,\alpha}\] and the \[\beta_i = \sum_d \beta_{i,d}.\] 

A code $C$ gives rise to a matroid, and in turn to a simplicial complex. We shall refer to the Stanley-Reisner ring of the code as $R(C)=R_{\Delta(H_C)}$.\\

We illustrate this by an example:
\begin{example}\label{exa5}
Let $C$ be the $[4,2,2]_2$ code defined by the generator matrix \[\begin{bmatrix} 1&1&0&0\\0&1&1&1\end{bmatrix}.\] A parity check matrix is \[\begin{bmatrix} 1&1&0&1\\0&0&1&1\end{bmatrix}.\] The set of bases of the associated matroid is \[\{\{1,3\},\{1,4\},\{2,3\},\{2,4\},\{3,4\}\},\] and the set of circuits (which in this case corresponds to the set of supports of non-zero codewords) is \[\{\{1,2\},\{1,3,4\},\{2,3,4\}\}.\] The Stanley-Reisner ring is therefore \[R(C) = \K[x_1,x_2,x_3,x_4]/<x_1x_2,x_1x_3x_4,x_2x_3x_4>.\]
A minimal free resolution of this ring is given by \[\begin{xymatrix}{0 & \ar[l] R(C) & S \ar[l] &&&\ar[lll]_{\tiny{\begin{bmatrix}x_1x_2 & x_1x_3x_4 & x_2x_3x_4\end{bmatrix}}} S^3 &&\ar[ll]_{\tiny{\begin{bmatrix}x_3x_4 &x_3x_4 \\ -x_2 & 0 \\0& -x_1\end{bmatrix}}} S^2 & \ar[l] 0}\end{xymatrix}.\] Using twists, we can rewrite it as  \[\begin{xymatrix}{0 & \ar[l] R(C) & S \ar[l] &\ar[l]_-{\tiny{\begin{bmatrix}1 & 1 & 1\end{bmatrix}}} S(-\{1,2\}) \oplus S(-\{1,3,4\}) \oplus S(-\{2,3,4\}) \\ \\&&0 \ar[r] &\ar[uu]_-{\tiny{\begin{bmatrix}1 &1 \\ -1 & 0 \\0& -1\end{bmatrix}}} S(-\{1,2,3,4\})^2 }\end{xymatrix}\] In the sequel, we will omit the maps between the modules. Note that, while the Betti numbers are unique, the maps are not. The $\N$-graded and ungraded Stanley-Reisner resolution of its associated matroid are \[0 \leftarrow  R(C) \leftarrow S \leftarrow S(-2) \oplus S(-3)^2 \leftarrow S(-4)^2 \leftarrow 0\] and \[0 \leftarrow R(C) \leftarrow S \leftarrow S^3 \leftarrow S^2 \leftarrow 0.\]
\end{example}

Our results are as follows. We start in Section~\ref{sec1} by giving two straightforward  statements (Proposition~\ref{prop1} and Corollary~\ref{cor2}) which enable us to conclude that a code is of constant weight using different assumptions than those in~\cite{Liu1}. 
In Section~\ref{sec3} we prove the main result of the paper, Theorem~\ref{maincor}, which partly builds on, and partly generalizes the result from~\cite{Liu1}. We determine the $\N$-graded Betti numbers of the Stanley-Reisner rings associated to the underlying matroid structures of constant weight codes. As shown in~\cite{Johnsen1}, we can derive the weight hierarchy of the code from its $\N$-graded Stanley-Reisner resolution. In particular we find that for constant weight codes these rings have pure (but not linear) resolutions. We also find a converse: Codes whose associated ring are of the given form are constant weight codes; in particular it is enough to find the first Betti number. At the end we show that the converse doesn't hold if we restrict ourselves to ungraded resolutions.

%
%

\section{The weight hierarchy of a constant weight linear code}\label{sec1}

 The weight hierarchies of constant weight codes were found in~\cite{Liu1}. There one proves the results by investigating value functions, and apply their results to a specific such value function. We will just restate it here, and refer to~\cite{Liu1}. Afterwards, we will give a converse, namely that a code with a given weight hierarchy is of constant weight. 

\begin{theorem}[{\cite[Theorem 1]{Liu1}}]\label{tdi} Let $C$ be a $k$-dimensional linear code over $\Fq$. Let $1 \leqslant s \leqslant k-1$. Suppose that all the $s$-dimensional linear subcodes of $C$ have the same weight $d_s$. Then for every $0\leqslant t \leqslant k$, and every linear subcode $D_t$ of dimension $t$ of $C$, we have \[w(D_t) = d_t = d_s\frac{ q^k-q^{k-t}}{q^k-q^{k-s}}.\]
\end{theorem}

This shows that being constant weight is the same as being constant weight in any dimension, except in dimension $0$ and dimension $k$.

\begin{corollary}\label{cdi} Let $C$ be a $k$-dimensional linear code over $\Fq$. Suppose that $C$ is of constant weight. Then the weight hierarchy $(d_1,...,d_k)$ is given by \[d_i=d\frac{q^i-1}{q^{i-1}(q-1)},\] where $d$ is the weight of any non-zero codeword.
\end{corollary}

\begin{example}\label{exawh}Let $C$ be the ternary code given by the generator matrix \[G= \begin{bmatrix}
1& 0& 1& 2& 0& 1& 2& 0& 1& 2& 0& 1& 2\\
0& 1& 1& 1& 0& 0& 0& 1& 1& 1& 2& 2& 2\\
0& 0& 0& 0& 1& 1& 1& 1& 1& 1& 1& 1& 1\end{bmatrix}\] This a constant weight code with weight $9$. Its weight hierarchy is \[(d_1,d_2,d_3)=(9,12,13).\]
\end{example}

The converse of this corollary is also true, namely, if a linear code has the weight hierarchy of a constant weight code, then it is itself a constant weight code. But there is an even stronger converse: 

\begin{proposition} \label{prop1}
 Let $C$ be a $[n,k,d]_q$-code. Assume that $d_k=\frac{q^k-1}{q^{k-i}(q^i-1)}d_i$ for some $1 \leqslant i < k$. Then $C$ is a constant weight code with weight $d_k\frac{q^{k-1}(q-1)}{q^k-1}.$
\end{proposition}

\begin{proof}

The proof is based on an easy corollary of lemma 1 in~\cite{Liu1}. We keep their notation. \[N_{r,1} m(PG(k-1,q)) \leqslant N_r \vartheta\] and there is equality if and only if all the $r$-dimensional projective subspaces have the same value $\vartheta$. Since \[d_i= d_k - max\{m(P_{k-i}),\ P_{k-i}\textrm{ is a }(k-i)\textrm{-dimensional projective subspace}\},\] this amounts to \[d_k \geqslant d_i \frac{q^{k-1}}{q^{k-i}(q^i-1)}\] with equality if and only if all the $(k-i)$-dimensional subspaces have the same value, or equivalently if and only if all the $i$-dimensional subcodes of $C$ have the same weight. The proposition then follows from Theorem 1 in~\cite{Liu1}.\qed
\end{proof}

\begin{corollary} \label{cor2}
Let $C$ be linear code over $\Fq$ of dimension $k$. Assume that there exists an integer $\alpha$ such that \[d_i=\alpha \frac{q^i-1}{q^{i-1}(q-1)}\  \forall 1\leqslant i \leqslant k.\] Then $C$ is constant weight, of weight $\alpha$.
\end{corollary}

\begin{remark}The Griesmer bound says that for a $[n,k,d]_q$ code, then $d_k \geqslant \sum_{i=0}^{k-1}\lceil \frac{d}{q^i} \rceil$. It is obvious that constant weight codes meet their Griesmer bound. The previous corollary could indicate that the converse is true. But it is not. Consider the $[5,2,3]_2$ code given by the generator matrix \[\begin{bmatrix} 1&1&1&0&0 \\ 0&0&1&1&1\end{bmatrix}\] This is not a constant weight code, but it reaches its Griesmer bound.
\end{remark}

%
%

\section{Betti numbers of the Stanley-Reisner resolution associated to a constant weight linear code}\label{sec3}

 We are now able to give the $\N$-graded resolution of a constant weight linear code. We use the notation of~\cite{Johnsen1}. As shown there, the $\N^n$-graded Betti numbers $\beta_{i,\sigma}$ of the Stanley-Reisner resolution of the matroid associated to the code are all zero, except for those subsets $\sigma$ of the ground set that are minimal (for the inclusion relation) such that $\# \sigma -rank(\sigma)=i$. We write \[N_i=\left\{\sigma \subset \{1,\cdots,n\},\ \#\sigma-rank(\sigma)=i\right\}.\]
Our first goals in this section are to show that if $\sigma \in N_i$, then $\sigma = Supp(C')$ for a subcode $C'$ of $C$ of dimension $i$, and then to prove that there is a one-to-one correspondence between subcodes and their supports. 

The first part of our plan is valid for any code:

\begin{lemma}\label{n=i} Let $C$ be a $[n,k,d]_q$ code. Let $0\leqslant i \leqslant k$ and $\sigma \in N_i$. Then there exists a subcode $C'$ of $C$ of dimension $i$ such that \[\sigma = Supp(C').\]
\end{lemma}
\begin{proof} Any circuit of the associated matroid is the support of a codeword. Namely a circuit is a minimal dependent subset of the columns of a parity check matrix, and this corresponds to a codeword (the converse is not true - see Example~\ref{1-1}). Since $\sigma \in N_i$, we know from~\cite{Johnsen1} that there exists a non-redundant set of $i$ circuits $\tau_1,\cdots,\tau_i$ such that \[\sigma= \bigcup_{j=1}^i \tau_j.\] As any circuit is the support of a codeword, we have found $i$ codewords $c_1,\cdots,c_i$ such that \[\sigma = \bigcup_ {j=1}^i Supp(c_j) = Supp(<c_1,\cdots,c_i>).\] It just remains to show that the subcode generated by these codewords is of dimension $i$. The non-redundancy property is the same as saying that there exists $i$ points $\{x_1,\cdots,x_i\}$ in $\{1, \cdots, n\}$ such that $x_l \in Supp(c_m)$ if and only if $l=m$. If $\sum_{j=1}^i \lambda_j c_j=0$ , then for every $1\leqslant l \leqslant i$, $\left.\left(\sum_{j=1}^i \lambda_jc_j\right)\right|_{x_l} = \lambda_l c_l|_{x_l} = 0 $ which implies that $\lambda_l=0$.\qed
\end{proof}

\begin{corollary} All elements in $N_i$ have the same cardinality $d_i$. The resolution is therefore pure.
\end{corollary}
\begin{proof} Theorem~\ref{tdi} shows that all the elements of $N_i$ have the same cardinality. The second part is~\cite{Johnsen1}. \qed
\end{proof}

The following is generally not valid for general codes, but it is for constant weight codes.

\begin{lemma}\label{lemma5}
Let $C$ be a constant weight $[n,k,d]_q$ code. Let $C'$ be a subcode and $c$ be a codeword. Then we have \[c \in C' \Leftrightarrow Supp(c) \subset Supp(C').\]
\end{lemma}
\begin{proof} One way is trivial. Assume now that $Supp(c) \subset Supp(C')$. Write $C'=<c_1,\cdots ,c_i>$ where the $c_j$'s are linearly independant. Let $x \in Supp(c)$. We can assume that $c|_x=1$. Consider the following codewords: \[c'_j=c_j-\left(c_j|_x\right)c,\] and the subcode $C''=<c'_1,\cdots,c'_i>$. It is obvious that $Supp(C'') \subset Supp(C')-\{x\}$. From Theorem~\ref{tdi}, we know that \[\#Supp(C') = d_i,\] and therefore $Supp(C'')<d_i$. Theorem~\ref{tdi} again shows that the dimension of the code $C''$ is strictly less that $i$, or equivalently that $c \in C'$.\qed
\end{proof}

\begin{proposition}\label{prop3} Let $C$ be a constant weight $[n,k,d]_q$ code. Then the mapping \[\begin{array}{ccc}\left\{\textrm{Subcodes of }C\right\} &\longrightarrow& 2^{\{1,\cdots,n\}} \\ C' & \longmapsto & Supp(C')\end{array}\] is injective.
\end{proposition}

\begin{proof}
Indeed, if $Supp(C') = Supp(C'')$, then any codeword of $C'$ is in $C''$ and vice versa.\qed
\end{proof}

\begin{example}\label{1-1} The converses of Lemma~\ref{lemma5} and Proposition~\ref{prop3} are not true. Consider namely the binary code given by the generator matrix \[\begin{bmatrix} 1&0&0&1&0\\0&1&0&1&0 \\ 0&1&1&0&1\end{bmatrix}.\] Then $c_1=(1,1,1,1,1)$ is a codeword whose support is $\{1,2,3,4,5\}$. The subcode generated by the codewords $c_2=(1,1,0,0,0)$ and $c_3=(0,0,1,1,1)$ has also support $\{1,2,3,4,5\}$. But  $c_2 \notin <c_1>$ and $<c_1>\neq<c_2,c_3>$. Moreover, even if $\{1,2,3,4,5\}$ is the support of a codeword, this is not a circuit in the associated matroid (since it contains a smaller dependent subset, for example $\{3,4,5\}=Supp(c_3)$.
\end{example}

\begin{proposition}\label{contribution} Let  $C$ be a constant weight $[n,k,d]_q$ code. Let $0\leqslant i \leqslant k$. then \[N_i=\left\{Supp(C'),\ C'\textrm{ is a subcode of dimension }i\right\}.\]
\end{proposition}
\begin{proof}One inclusion is Lemma~\ref{n=i}. Let now $C'$ be a subcode of dimension $i$. Then by Theorem~\ref{tdi}, we know that $\#Supp(C') = d_i$. If it wasn't in $N_i$, then there would exists a subset $X \subsetneq Supp(C')$ such that $X \in N_i$. By Lemma~\ref{n=i} again, we would find a subcode $C''$ of dimension $i$ such that $Supp(C'')=X$. But again by Theorem~\ref{tdi}, we would get that \[d_i = \#Supp(C'') = \#X < \#Supp(C') = d_i\] which is absurd.\qed
\end{proof}

From Proposition~\ref{contribution} we know that the non-zero contributions to the term of homological degree $i$ in the Stanley-Reisner resolution of the matroid $\Delta$ associated to a constant weight linear code $C$ come from its subcodes of dimension $i$. We also know (\cite[Hochster's formula]{1Miller}) that  

\[\beta_{i,\sigma} = \tilde{h}_{|\sigma|-i-1}(\Delta|_\sigma,\K).\] 

Let $\Delta '$ be the simplicial complex where the facets are the independent sets of the matroid $\Delta(H_{C'})$, for $H_{C'}$ a parity check matrix of $C'$ (for example obtainable by adding an appropriate number of rows to $H_{C'}$). 

\begin{lemma} \label{indep}
Let $C'$ be a subcode of a linear code $C$ of constant weight. If $\sigma \subset Supp(C')$, then $\Delta '|_\sigma=\Delta|_\sigma$.
In particular \[\beta_{i,\sigma}(R(C')) = \beta_{i,\sigma}(R(C)).\]
\end{lemma}
\begin{proof}
Clearly, if some columns of $H_{C}$ indexed by a subset $\tau$ of $\sigma$ are independent, then the corresponding columns of $H_{C'}$ are independent. If, on the other hands the columns of $H_{C}$ indexed by such a $\tau$ are dependent, then there is a codeword $c \in C$ with support inside $\tau \subset \sigma \subset Supp(C')$.
By Lemma~\ref{lemma5} we then have $c \in C'.$ Hence the columns indexed by $\tau$ are dependent in $H_{C'}$ also. Hence the lemma holds.\qed
\end{proof}

\begin{example}The previous lemma doesn't necessarily hold if the code is not constant weight. For the code given in Example~\ref{1-1}, the matroid $\Delta$ associated to it has bases set \[\{\{1,5\},\{2,5\},\{1,3\},\{3,5\},\{2,3\},\{3,4\},\{4,5\}\}\] while the subcode generated by $c_2,c_3$ has a associated matroid $\Delta'$ with bases set \[\{\{1,3,5\},\{2,3,4\},\{1,4,5\},\{2,4,5\},\{2,3,5\},\{1,3,4\}\}.\] Take $\sigma=\{1,2,3,4,5\}$. Then $\Delta|_\sigma = \Delta \neq \Delta' = \Delta'|_\sigma.$
\end{example}

Before proving our main theorem, we need a combinatorial relation between the number of subcodes of a given dimension. A Grassmannian is a space that parametrizes all the linear subspaces of a given dimension of a vector space. Translated to coding theory, a Grassmannian is a space that parametrizes all the linear subcodes of a given dimension of a linear code. The number of linear subspaces (alt. subcodes) of dimension $r$ of a vector space (alt. code) of dimension $k$ over $\Fq$ is given by \[{k \brack r}_q = \frac{f(k,q)}{f(r,q)f(k-r,q)}\] where $f(n,q) = \prod_{i=1}^{n}(q^i-1).$
\begin{lemma}\label{gauss} Let $k \geqslant 0$. Then \[\sum_{i=0}^{k-1}(-1)^{k+i-1}{k \brack i}_q q^\frac{i(i-1)}{2} = q^\frac{k(k-1)}{2}.\]
\end{lemma}
\begin{proof}The result is obtained by taking $t=-1$ in Gauss binomial theorem (\cite{Konvalina1}): \[\sum_{i=0}^k{k \brack i}_qq^{\frac{i(i-1)}{2}}t^i = \prod_{i=0}^{k-1}(1+q^it).\]\qed
\end{proof}

We can now prove the main theorem of this section, namely a description of the Stanley-Reisner resolution of the matroid associated to a constant weight code.

\begin{theorem} \label{maincor}
Let $C$ be be a constant weight $[n,k,d]_q$ code. Then the $\N$-graded Stanley Reisner resolution of the matroid associated to the code is \[ 0 \leftarrow R(C) \leftarrow S \leftarrow \cdots \leftarrow S(-d_i)^{{k \brack i}_qq^\frac{i(i-1)}{2}} \leftarrow \cdots \leftarrow 0,\] where $d_i=d\frac{q^i-1}{q^{i-1}(q-1)}$ for $1\leqslant i \leqslant k$.
\end{theorem}
\begin{proof}We do it recursively on the dimension $k$ of the code.  For $k=1$, all the non-zero codewords have the same support, say $a \subset \{1,...,n\}$, so that $R(C)$ is of the form $R(C)=\K[\mathbf{x}]/<\mathbf{x}^a>$ and the Stanley-Reisner resolution is \[0 \leftarrow R(C) \leftarrow S \leftarrow S(-a) \leftarrow 0.\] Suppose that we have proved our result for all constant weight codes of dimension less than $k$. In particular, for any constant weight code $C'$ of dimension $i \leqslant k$, $\beta_{i,Supp(C')}(R(C'))=q^\frac{i(i-1)}{2}.$ By Proposition~\ref{contribution} and~\cite{Johnsen1}, we know that the $\N^n$-graded Stanley-Reisner resolution of the code is \begin{eqnarray*} && 0 \leftarrow R(C) \leftarrow S \leftarrow \bigoplus_{C' \in \mathcal{C}_1(C)}S(-Supp(C'))^{\beta_{1,Supp(C')}(R(C))} \leftarrow \cdots \\&&\leftarrow \bigoplus_{C'\in \mathcal{C}_{k-1}(C)}S(-Supp(C'))^{\beta_{k-1,Supp(C')}(R(C))} \leftarrow S(-Supp(C))^{\beta_{k,Supp(C)}(R(C))} \leftarrow 0.\end{eqnarray*}By Lemma~\ref{indep} and the recursion hypothesis, we can assert that the Stanley-Reisner resolution is \begin{eqnarray*} && 0 \leftarrow R(C) \leftarrow S \leftarrow \bigoplus_{C' \in \mathcal{C}_1(C)}S(-Supp(C')) \leftarrow \cdots \leftarrow \bigoplus_{C'\in \mathcal{C}_{i}(C)}S(-Supp(C'))^{q^\frac{i(i-1)}{2}} \leftarrow\cdots \\&&\leftarrow \bigoplus_{C'\in \mathcal{C}_{k-1}(C)}S(-Supp(C'))^{q^\frac{(k-1)(k-2)}{2}} \leftarrow S(-Supp(C))^{\beta_{k,Supp(C)}(R(C))} \leftarrow 0.\end{eqnarray*}
Since there are exactly ${k \brack i}_q$ subcodes of dimension $i$, it gives that the ungraded resolution is \[ 0 \leftarrow R(C) \leftarrow S \leftarrow S^{{k \brack 1}_qq^\frac{1(1-1)}{2}} \leftarrow \cdots \leftarrow S^{{k \brack k-1}_qq^{(k-1)(k-2)}{2}} \leftarrow S^{\beta_{k,Supp(C)}} \leftarrow 0.\] 
In this resolution we study the terms of degree $\#Supp(C)$ in the Hilbert polynomials of each of the terms. The alternating sum is zero (as is the contribution from $R(C)$), so \[0 = \left(\sum_{i=0}^{k-1}(-1)^i{k \brack i}_qq^\frac{i(i-1)}{2} \right) + (-1)^k \beta_{Supp(C),k}\] and from Lemma~\ref{gauss}, this gives \[\beta_{Supp(C),k}=q^\frac{k(k-1)}{2}.\]
\qed\end{proof}

\begin{example}Take the same code as in Example~\ref{exawh}. Then the Stanley-Reisner resolution of the matroid associated to $C$ is \[0 \leftarrow R(C) \leftarrow S \leftarrow S(-9)^{13} \leftarrow S(-12)^{39} \leftarrow S(-13)^{27} \leftarrow 0.\]
\end{example}

Of course, since the $\N$-graded Stanley-Reisner resolution gives the weight hierarchy, the converse of the previous corollary is true. But there is a stronger converse:

\begin{proposition}Let $C$ be a $[n,k,d]_q$ linear code. Suppose that the Stanley-Reisner resolution of its associated matroids starts like \[0 \leftarrow R(C) \leftarrow S \leftarrow S(-d)^{{k \brack 1}_q} \leftarrow \cdots\] Then $C$ is a constant weight code of weight $d$.
\end{proposition}

\begin{proof} Since in homology degree $1$, the contribution of any subset $\sigma$ of the matroids ground set is $1$ if $\sigma$ is a circuit, and $0$ otherwise, the start of the resolution tells us that there are exactly ${k \brack 1}_q$ circuits of weight $d$. We know that any circuit of the matroid corresponds to a vector space generated by a codeword. So this tells us that there are at least ${k \brack 1}_q$ subspaces generated by a single codeword. But there are ${k \brack 1}_q$
 subspaces of dimension $1$, which means that all the subspaces of dimension $1$ are generated by a codeword of weight $d$.\qed
\end{proof}

We know the $\N$-graded Stanley-Reisner resolution of a constant weight linear code. As such, we also know the ungraded Stanley-Reisner resolution (just remove the twists since this is a pure resolution). A natural question would be to determine whether a code with such a ungraded Stanley-Reisner resolution is constant weight. The answer is no, as the following example shows.

\begin{example} Let $C$ be the code of Example~\ref{exa5}. Its ungraded Stanley-Reisner resolution is the same as the ungraded Stanley-Reisner resolution associated to the $[4,2,2]_2$ constant weight code defined by the generator matrix \[\begin{bmatrix} 1&0&1&0 \\ 0&1&1&0\end{bmatrix}\]
\end{example}

\subsubsection*{Acknowledgements} We thank the anonymous referees for helpful remarks.

%
%


\begin{thebibliography}{1}

\bibitem{Bjorner1}Bj\"orner, A.: \textit{The homology and shellability of matroids and geometric lattices}. In: Matroid Applications. Encyclopedia of Mathematics Application, vol. 40, pp 226--283. Cambridge University Press, Cambridge (1992).

\bibitem{Eagon1}Eagon, J.A., Reiner, V.: \textit{Resolutions of Stanley-Reisner rings and Alexander duality}. J. Pure Appl. Algebra \textbf{130}(3), 265--275 (1998).

\bibitem{Helleseth1}Helleseth, T., Kl\o ve, T., Mykkeltveit, J.: \textit{The weight distribution of irreducible cyclic codes with block lengths $n_1((q^l-1)/N)$}. Discr. Math. \textbf{18}, 179--211 (1977).


\bibitem{Johnsen1} Johnsen, T., Verdure, H.: \textit{Hamming weights and Betti numbers of Stanley-Reisner rings associated to matroids}. Appl. Algebra Eng. Comm. Comput., http://link.springer.com/article/10.1007/s00200-012-0183-7, to appear. arXiv:1108.3172 (2011)

\bibitem{Konvalina1} Konvalina, J.:  \textit{A unified interpretation of the binomial coefficients, the Stirling numbers, and the Gaussian coefficients}. Am. Math. Month. \textbf{107}(10), 901--910 (2000). 

\bibitem{Liu1} Liu, Z., Chen, W.: \textit{Notes on the value function}, Des. Codes Cryptogr. \textbf{54}, 11--19 (2010).

\bibitem{Martin1} Martin, J.: \textit{Matroids, demi-matroids and chains of linear codes}, Master thesis. http://hdl.handle.net/10037/2957 (2010).

\bibitem{1Miller} Miller, E., Sturmfels, B.: \textit{Combinatorial Commutative Algebra}, GTM 227. Springer, New York (2005).

\bibitem{1Oxley} Oxley, J.G.: \textit{Matroid theory}, Oxford University Press, Oxford (1992).

\end{thebibliography}
\end{document}